\documentclass{llncs}
\usepackage{amsmath,amsfonts,amssymb,amscd}
\usepackage{enumerate}
\usepackage{mathrsfs}
\usepackage[x11names]{xcolor}
\usepackage{graphicx}
\usepackage{booktabs}       %
\usepackage{listings}
\usepackage[camera]{dtrt}
\usepackage{framed}
\usepackage[operators,probability,sets]{cryptocode}
\usepackage{mdframed}
\usepackage{xspace}
\usepackage{tikz}
\usepackage{upgreek}
\usepackage{enumitem}

\usepackage[caption=false]{subfig}
\usepackage{etoolbox}
\AtEndEnvironment{proof}{\phantom{}\qed}

\newenvironment{frameditemize}[1] 
{\begin{mdframed}[linewidth=1pt, frametitle={#1},frametitlealignment=\centering]
\begin{itemize}[leftmargin=*]}
{\end{itemize}
\end{mdframed}}
\newenvironment{myproof}{\begin{proof}}{\hfill \end{proof}}

\newtoggle{anonymous}
\newtoggle{comments}
\newtoggle{short}
\togglefalse{anonymous}
\toggletrue{comments}
\toggletrue{short}

\newcommand{\wenhao}[1]{\dtnote[WW]{#1}}

\newcommand{\fanz}[1]{\dtcolornote[FZ]{blue}{#1}}
\newcommand{\aviv}[1]{\dtnote[AY]{#1}}
\newcommand{\tool}{$Proo\upvarphi$\xspace}

\newcommand{\capacity}{{s}}
\newcommand{\txfee}{{f}}
\newcommand{\cost}{{p}}

\newcommand{\numprover}{N}
\newcommand{\numuser}{n}

\newcommand{\capsum}{\bar{S}}

\usepackage[style=numeric,backend=biber,maxbibnames=2]{biblatex}
\usepackage[unicode,naturalnames]{hyperref}
\usepackage[capitalize]{cleveref}

\makeatletter
\def\toclevel@title{-1}
\def\toclevel@author{0}
\makeatother
\title{\texorpdfstring{\tool}{Proofee}: A ZKP Market Mechanism}
\author{
Wenhao Wang\inst{1},
Lulu Zhou\inst{1},
Aviv Yaish\inst{1},
Fan Zhang\inst{1},
Ben Fisch\inst{1}, and
Benjamin Livshits\inst{2}
}
\institute{
Yale University, New Haven, CT, USA,
\and
Imperial College London, London, UK
}

\begin{document}
\maketitle
\begin{abstract}

Zero-knowledge proofs (ZKPs) are computationally demanding to generate.
Their importance for applications like ZK-Rollups has prompted some to outsource ZKP generation to a market of specialized provers.
However, existing market designs either do not fit the ZKP setting or lack formal description and analysis.

In this work, we propose a formal ZKP market model that captures the interactions between users submitting ZKP tasks and provers competing to generate proofs.
Building on this model, we introduce \tool, an auction-based ZKP market mechanism.
We prove that \tool is incentive compatible for users and provers, and budget balanced.
We augment \tool with system-level designs to address the practical challenges of our setting, such as Sybil attacks, misreporting of prover capacity, and collusion.
We analyze our system-level designs and show how they can mitigate the various security concerns.
\end{abstract}

\section{Introduction}
\label{sec:intro}

Zero-knowledge proofs (ZKPs) enable efficient verification of computation and are used by blockchain scalability solutions (e.g., ZK-Rollups~\cite{zksync,scrollTFM,starknet}), user authentication~\cite{baldimtsi2024zklogin}, data oracles~\cite{zhang2020deco}, and many more.
Despite recent efficiency improvements, generating ZKPs remains computationally expensive~\cite{ernstberger2024zk}, often requiring specialized infrastructure (e.g., GPU or ASIC).
This has naturally led to the emergence of \emph{ZKP markets}~\cite{nil,gevuloteco,roy2024succinct} that allow users to outsource proof generation to specialized provers.
Ideally, an effective market will not only improve user experience but also lower user costs by fostering competition among provers.
However, since such markets operate in an open and decentralized environment, challenges arise because users and provers might be malicious.

Market designs by both industry and academia~\cite{Taiko,gevuloteco,nil,scrollTFM} are limited.
For example, the elegant work of \citeauthor{gong2024v3rified}~\cite{gong2024v3rified} considers a market comprising one user and multiple provers.
However, proving multiple user transactions in batches is a key performance optimization in ZK-Rollups.
Moreover, ZKP applications, such as ZK-Rollups, are already serving a substantial user base~\cite{l2beat-activity}, highlighting the need for mechanisms that address multi-user scenarios.
On the other hand, the mechanism advanced by the commercial prover market  Gevulot~\cite{gevuloteco} supports multiple users but relies on a fixed fee for all tasks, implying users willing to pay more may have to wait longer than those willing to pay less.
A survey of the literature (see~\cref{sec:related}) shows that other proposed mechanisms either do not apply to our setting or lack formal specifications and analysis.

\parhead{This work}
We formally model ZKP markets and dissect several designs.
In particular, we propose \tool, a ZKP market mechanism that comprises a \textit{core auction mechanism} and \textit{system-level designs}.
At a high level, the core mechanism of \tool runs an auction between users and provers and allocates a set of low-cost provers with another set of high-value user tasks.
We formally analyze the properties of the core mechanism and show that it is budget-balanced and guarantees incentive compatibility of market participants.
\wenhao{rephrased here}
We then introduce system-level designs to enhance the core mechanism and mitigate security threats that the core auction mechanism alone cannot address, and provide an analysis of how these designs defend against untruthful capacity bids, Sybil attacks, and collusion.
The main challenges lie in modeling the ZKP market to capture the setup of practical systems, achieving desired security properties through game-theoretic designs (e.g., with auctions) and system-level designs (e.g., using cryptography), and developing rigorous formal analysis.

\subsection{ZKP Market Model}
We consider a model comprising users, who have ZKP generation tasks, and provers, who finish user tasks for rewards.
The market operates in rounds of auctions, and each auction handles a specific {\em type} of tasks, i.e., tasks of the same circuit and ZKP scheme.
In each auction, users submit tasks and specify a fee $\txfee$ for each task.
The fee reflects how much the user values this task, also called the {\em value} of a task.
Meanwhile, each prover specifies a capacity $\capacity$, i.e., the number of tasks they can handle within a predefined time frame (tailored to application needs), and their unit cost $\cost$ i.e., the cost to finish one given task.
Given user tasks and provers' capacities and unit costs, a market mechanism decides \textit{allocation} and \textit{payments}: it selects a subset of tasks that are to be proven, and for each one, the prover that would generate the associated proof, the user's payment, and the amount that can be collected as revenue by the prover. Note that user payments and the prover's revenue may differ from user-chosen fees.

Our model explicitly captures the fact that practical provers produce proofs in batch, by allocating user tasks in batches.
Assigning tasks of the same type to one prover is more efficient than assigning them to different provers, as it avoids context-switching costs (e.g., loading proving keys and circuits to the memory)
and can leverage efficient batch proofs available in some ZKP schemes~\cite{kate2010constant,chen2023hyperplonk,gabizon2019plonk}.
Although we use ZK-Rollup terminologies, our market model applies to other systems that include a market of verifiable services, such as zkLogin~\cite{baldimtsi2024zklogin}, where users need to generate ZKPs to authenticate their identity.
\done\fanz{include zkLogin example}

\subsection{Our Core Market Mechanism: \tool}
\label{sec:core-mechanism}
Inspired by second-price auction mechanisms such as VCG~\cite{vickrey1961counterspeculation} and classic double auction mechanisms~\cite{mcafee1992dominant}, we present our core ZKP market mechanism named \tool.
In this mechanism, we guarantee incentive compatibility by paying each allocated prover the ``second price'', i.e., the price reported by the unallocated prover with the lowest cost, and ensuring each allocated prover's capacity is used in full.
Concretely, the mechanism first hypothetically allocates the highest-paying user tasks to the lowest-cost provers and stops when reaching a task with a fee that no longer covers the cost of the prover.
The mechanism then selects this hypothetical set of provers excluding the one with the highest unit cost, and selects the highest value tasks up to the capacity of the selected provers.
The core mechanism of \tool is specified below.

\medskip
\begin{frameditemize}{Core Mechanism of \tool}
    \item The mechanism collects user bids $\{\txfee_i\}_{i = 1}^{\numuser}$, and prover bids $\{ (\capacity_j, \cost_j) \}_{j = 1}^{\numprover}$ with $\txfee_1 \geq \cdots \geq \txfee_\numuser$ and $\cost_1 \leq \cdots \leq \cost_\numprover$.
    Let $\capsum_j := \sum_{i = 1}^{j} \capacity_i$ denote the sum of the first $j$ provers' capacities.
    \item The mechanism determines the largest $\ell$ such that $\cost_{j + 1} \leq \txfee_{(\capsum_j + 1)}$, i.e. $\ell = \underset{j}{\arg\max} \{ \cost_{j + 1} \leq \txfee_{(\capsum_j + 1)} \}$.
    The first $\ell$ provers are allocated in full, with the first $\capsum_\ell$ user tasks.
    \item
    Each selected user is charged $\txfee_{\capsum_\ell + 1}$.
    Each selected prover with index $j \in [1, \ell]$ is $\capacity_j \cdot \cost_{\ell + 1}$.
\end{frameditemize}

To provide intuition, we run the mechanism through a simple example.
\begin{example}
\label{eg:brief-example-intro}
Suppose a market with $8$ tasks with values \done\fanz{make sure we use fees and values consitently} $(10,10,10,10,9,9,1,1)$, and $3$ provers with capacities $s_i$ and costs $p_i$ such that $(\capacity_1, \cost_1) = (4, 0), (\capacity_2, \cost_2) = (2, 1), (\capacity_3, \cost_3) = (2, 10)$.
When all parties are bidding honestly, we have $\ell = 1$.
The first $\capsum_\ell = \sum_{j =1}^\ell s_j = 4$ task are allocated, and are charged $f_{\capsum_\ell+ 1} = 9$.
Prover $1$ is paid $p_{\ell+1}=1$ per task, and has utility $(1 - 0) \times 4 = 4$.
\end{example}

One interpretation is that we first greedily match high-paying tasks with low-cost provers until none are left.
E.g., $\cost_1$ is matched $\txfee_1$ through $\txfee_4$, and so on.
Then, we compare a prover $k$'s cost $\cost_k$ with the highest fee of her matched tasks (which is precisely $\txfee_{\capsum_{k-1}+1}$). E.g., compare $\cost_1$ with $\txfee_1$, $\cost_2$ with $\txfee_5$, and $\cost_3$ with $\txfee_7$.
Let's call a prover {\em feasible} if $\cost_k \le \txfee_{\capsum_{k-1}+1}$.
Finally, the mechanism allocates {\em all but the last} feasible provers to their full capacity and calls the number of allocated provers $\ell$.
In this example, $\cost_1$ and $\cost_2$ are feasible, so $\ell=1$.
The last feasible prover's cost and its highest-paying task's fee set the prices for provers and users, respectively. I.e., provers are paid $\cost_{\ell+1}$ and users are charged $f_{\capsum_\ell +1}$.

We prove \tool achieves several desirable properties.
It is {\em budget-balanced} (\cref{prop:new-wbb}), i.e., the payments to provers are always covered by the fees collected from users.
Moreover, it is \textit{incentive compatible} for users and provers to bid honestly (\cref{prop:new-duic,prop:new-dpic}).

\subsection{Implementing \texorpdfstring{\tool}{Prooϕ} in an Open ZKP market}
The core mechanism alone leaves certain security threats open, thus we enhance it with system-level solutions, including slashing, fixing prover capacity bids over multiple rounds, and encrypting bids to hide information required for profitable deviations such as misreporting of prover capacity, Sybil attacks, and collusion.
{When considering risk-averse actors (in line with previous work making the same assumption~\cite{fisch2017socially,cong2020decentralized,roughgarden2021ignore,chung2023foundations}), our solutions thwart possible risk-free threats.}

\parhead{Incorrect or missing proofs}
To ensure the correctness and timeliness of the completion of the ZKP tasks, we require provers to deposit collateral when joining the protocol,
and seize the collateral of provers who generate incorrect ZKPs or do not generate them in a timely manner (as detailed in \cref{subsec:missing-zkp}).

\parhead{Misreporting capacity}
The core \tool mechanism is incentive compatible for provers, meaning that they truthfully report their \textit{proving costs}.
However, provers may profitably deviate in other ways, such as by misreporting their \textit{capacity}.
In~\cref{prop:capacity-misreport-fail}, we prove that such deviations are worthwhile only when the bids of other users and provers satisfy a narrow condition, implying that dishonest provers would have to monitor user bids and adjust their reported capacities correspondingly.
Following this observation and given that a prover's capacity is not likely to change drastically in a short amount of time, the threat can be mitigated by restricting the changes of the capacity bid over time; for example, a prover can only change the capacity every $n$ rounds by $m$ folds (where $n,m$ are parameters, possibly adjusted via, e.g., community votes in a DAO).

\parhead{Sybil attacks}
A malicious prover could pose as multiple actors (provers or users) and mount {\em Sybil attacks}: a prover may create several Sybil provers or submit ``fake'' user tasks \cite{gafni2023optimal}.
In~\cref{prop:sybil-good}, we prove that the threat of such attacks is limited.
Particularly, a prover splitting into multiple provers never harms the social welfare if such an attack is profitable.
In the other case, where a prover generates fake user bids to gain more profit, we show that such attacks are only profitable when the bids of the other users and provers satisfy certain conditions.
We show in~\cref{prop:prover-sybil-user} that if the prover has no information on the bids submitted by other parties, there always exists a possible configuration of these bids that would result in the prover having a loss.
Therefore, this type of Sybil attack can be mitigated by concealing user and prover bids from the provers, i.e., making the auctions sealed-bid.

\parhead{Collusion.}
Collusion can occur among different actors, e.g., users and provers or among a group of provers.
In~\cref{prop:collusion-prover-user,prop:collusion-two-prover}, we show that for the collusion to be definitely profitable, the colluders need to have full information of the other parties' bids.
Therefore, these types of collusion can be intuitively mitigated using a similar approach as we use for Sybil attacks, i.e., hiding the bids of other parties from provers.
Still, our protocol cannot properly circumvent all collusion (e.g., all provers can collude with each other and form a monopoly), but it is common for mechanism designs to disregard collusion~\cite{gong2024v3rified,huang2002design}.

\section{Related Work}
\label{sec:related}

\parhead{Commercial mechanisms for ZKP markets}
Several commercial platforms suggested ZKP market mechanisms without formal analysis.
Gevulot~\cite{gevuloteco} uses posted prices to determine the set of provers (i.e., the same predetermined fee per task is paid to selected provers);
Taiko~\cite{Taiko} lets provers set their prices; %
Scroll~\cite{scrollTFM} selects and partially subsidizes the lowest-cost provers and pays them with their bid\wenhao{TODO: add webarchive}; 
=nil;~\cite{nil} facilitates price discovery by maintaining a limit order book for each circuit with buy orders from users and sell orders from provers.
An order is executed when the buying and selling prices meet.
In \cref{tab:comparison}, we list the protocols across three dimensions:
how provers are selected (``Prover Selection''),
how user payments are determined (``User Payment''), and how the profit of provers is decided (``Payment to Provers'').

\begin{table*}[tb] 
\centering\footnotesize
\caption{Summary of prover market designs.} 
\label{tab:comparison}
\begin{tabular}{llll} 
\toprule
\textbf{Protocol}& \textbf{User Payment} & \textbf{Prover Selection} & \textbf{Payment to Provers} \\ \midrule
\tool & second price & lowest price & second price $\times$ allocated capacity   \\ 
=nil; & limit order book & limit order book & limit order book \\ %
Taiko & posted price & random selection & prover-specified payment \\ %
Scroll & first price & random selection & partially subsidized \\ %
Gevulot & posted price & random selection & posted price $\times$ allocated capacity \\
\bottomrule

\end{tabular}
\end{table*}

\parhead{Other computation outsourcing mechanisms}
\citeauthor{thyagarajan2021opensquare}~\cite{thyagarajan2021opensquare} propose OpenSquare, a Verifiable Delay Function (VDF) market that incentivizes maximum server participation, which is resistant to censorship and single-point failures.
In a broader sense, ZKP markets are computation outsourcing markets.
In V3rified~\cite{gong2024v3rified}, the fee mechanism for computational tasks is categorized into revelation ones (i.e., an auction where provers bid their costs) and non-revelation ones (the client posts its task with a given fee) and characterize their power and limitations.
The protocols in V3rified cannot be directly applied to our model, as they are tailored for multiple provers proving one user task.

\parhead{Double auctions}
Our setting is of a two-sided market: on the demand side, users submit transactions that require proving, while on the supply side, provers provide proving capacity.
Double auctions are commonly used to coordinate trade in such markets, i.e., to choose which agents get to trade and at what prices.
While there is some overlap between the ``traditional'' setting explored by auction theory literature and the ZKP market setting, we note that the latter introduces new challenges: provers may collude, either amongst themselves or with users, and provers have a capacity parameter to bid.
\citeauthor{mcafee1992dominant}~\cite{mcafee1992dominant} presents a mechanism that is incentive compatible, individually rational, and budget-balanced for unit-demand buyers and unit-supply sellers (i.e., each buyer wishes to purchase a single item, and each seller offers a single item for sale). %
\citeauthor{huang2002design}~\cite{huang2002design} show a mechanism that is also incentive compatible, individually rational, and budget-balanced for multi-unit buyers and sellers (i.e., buyers may want multiple items, and suppliers may sell multiple items). %
However, their work assumes public capacity information and no collusion, and they are not a direct generalization of our mechanism to multi-item buyers.
Although outside the scope of our work, we refer readers interested in a broad review of the literature to the survey of \citeauthor{parsons2011auctions}~\cite{parsons2011auctions}, which covers a variety of auction formats, including double auctions.

\section{Model}
\label{sec:modeling}

\subsection{ZKP Market}
In a ZKP market, user-specified ZKP generation tasks are outsourced to specialized parties (i.e., provers) for a price.
The market is specified by a core mechanism and a system-level protocol.
We now elaborate on these components.

\parhead{Roles}
We use $\numuser$ to denote the number of users in the market, and each user has a ZKP generation task.
Each user task has a value $\txfee$ that the user is willing to pay to complete the task.
There are $\numprover$ provers in the market willing to complete ZKP generation tasks for rewards.
Each prover has a ZKP generation capacity $\capacity$, i.e., the number of tasks it can complete within a fixed time window, and a unit cost $\cost$, i.e., its cost to complete one ZKP task.
Note that in our model, we assume that each prover's total cost is the sum of the costs of all its tasks.
Besides users and provers, a coordinator (or auctioneer) is responsible for collecting bids from users and provers and executing the market system and mechanism.
For instance, when applied to ZK-Rollups, the centralized sequencer can be the auctioneer.

\parhead{Core market mechanism}
In each round, the market mechanism selects a set of user tasks that can be proven and an allocation of the tasks to the provers.
Additionally, the core mechanism determines the amount of fee that needs to be collected from each user and determines the payment to each prover.

\parhead{Market system}
The market system specifies the additional steps and requirements for users and provers during the execution of the market mechanism.
An example is that the provers deposit some collateral before they are eligible to submit bids in the market.
The goal of the market system is to work in synergy with the market mechanism to secure the market against security threats.

\ignore{\parhead{Assumptions}
We further make the following assumptions for our analysis.
\begin{itemize}
    \item All parties are myopic, i.e., they are only concerned with their utility within one batch of tasks.
    \item Each user has exactly one task. Note that this is a widely used assumption \XXX \wenhao{TODO: cite a bunch of papers}
    \item If a user task is not included in a batch, its utility is zero; if a prover is not allocated to prove anything, its utility is zero.
    \fanz{1) this repeats the ``role'' paragraph; 2) I don't understand the argument here. regardless of whether a prover is fully allocated, its total cost can scale linearly or sublinearly, depending on how we model it. These two seem orthogonal.}
\end{itemize}}

\parhead{Utilities}
Here we introduce the utility of the market participants.
We assume that all parties are myopic, i.e., they are only concerned with their utility within one round of market.
We further assume that the auctioneer is trusted and therefore does not have a utility to maximize.
For a user whose valuation of its task is $\txfee$ and pays $\txfee'$, if the task is selected, its utility is $\txfee- \txfee'$; otherwise, if the task is not selected, its utility is $- \txfee'$.
For a prover whose unit cost is $\cost$, and is paid $w$ and allocated $\capacity'$ tasks, then its utility is $w - \capacity' \cdot \cost$.

\subsection{Threat Model}
Both users and provers can act strategically and deviate from the protocol arbitrarily.
The attacks on the market may occur in the following forms.

\parhead{Incorrect or missing ZKPs}
Provers may fail to generate correct ZKPs in the allotted time, whether intentionally or otherwise.

\parhead{Misreporting bids}
Both users and provers may not bid truthfully. 
If any actor can benefit from misreporting its bid, we consider this as an attack.

\parhead{Sybil attacks}
Besides misreporting bids, a user or a prover may pose as multiple actors and submit fake bids.
Sybil attacks can occur when a prover submits bids as fake users or a prover submits bids as fake provers.

\parhead{Collusion}
Users and provers may collude to bid strategically to increase their joint utility.
Not all forms of collisions harm players outside the coalition, but those that do are considered a security threat.

{
\begin{remark}
\label{rem:risk}
Previous work on TFMs typically adopted a non-Bayesian perspective, including when analyzing different possible threats, i.e., misreporting values, Sybil attacks, and collusion \cite{chung2023foundations,gafni2024discrete,gafni2024barriers,bahrani2024transaction,roughgarden2021transaction}.
In practice, blockchain actors such as Bitcoin miners who invest resources to receive rewards may prefer to minimize the financial risk entailed in their operation, with prior work seeing this as an explanation for the popularity of mining pools.
Therefore, we follow previous work that modeled actors as risk-averse~\cite{fisch2017socially,cong2020decentralized,roughgarden2021ignore,yaish2023correct}, who will deviate from the honest behavior only if it would certainly improve their payoffs~\cite{chung2023foundations}.
We note that alternative models could lead to interesting future work (see \cref{sec:conclusions}).
\end{remark}
}

\subsection{Desiderata}
\label{subsec:desiderata}

Following prior work~\cite{roughgarden2021transaction,mcafee1992dominant}, we devise a desiderata for ZKP market mechanisms.

\parhead{Budget balance}
In a ZKP market, the fees collected from users in each round should cover the payment to provers, a property called {\em budget balance}.
\begin{definition}[Budget Balance (BB)]
\label{def:budget-balance}
A mechanism is {\em budget-balanced} if the fees collected from the users are no less than the sum of payment to provers.
\end{definition}

\parhead{Incentive compatibility}
Our mechanism should incentivize actors to truthfully report their values, i.e., it should satisfy \cref{def:dsic}.
In \cref{def:uic,def:pic} we provide the corresponding definitions for users and provers, respectively.

\begin{definition}[Dominant Strategy Incentive Compatibility (DSIC)]
\label{def:dsic}
    A mechanism is dominant strategy incentive compatible if it is always best for each participant to bid its true valuation.
\end{definition}

\begin{definition}[User DSIC (UDSIC)]
\label{def:uic}
A mechanism UDSIC if the mechanism is DSIC for users.
\end{definition}

\begin{definition}[Prover DSIC (PDSIC)]
    \label{def:pic}
    A mechanism is PDSIC if it is DSIC for provers to bid their costs honestly.
\end{definition}
Here we note that with PDSIC, it is still possible that a prover may misreport its capacity and get extra utility.

\parhead{Sybil proofness}
Provers submitting bids under fake identities should not negatively impact social welfare if creating Sybils is profitable.
Recall that provers can bid under fake prover or user identities.
We define them respectively as {\em prover Sybil proofness} (\cref{def:prover-sybil-proof}) and {\em user Sybil proofness} (\cref{def:user-sybil-proof}).
\begin{definition}[Prover Sybil Proofness]
\label{def:prover-sybil-proof}
A mechanism is prover-Sybil-proof if other parties' utilities are not reduced when any prover profits from creating prover Sybils.
\end{definition}

\begin{definition}[User Sybil Proofness]
\label{def:user-sybil-proof}
A mechanism is user-Sybil-proof if others' utilities are not reduced when any prover profits from creating user Sybils.
\end{definition}

\parhead{Collusion resistance}
A ZKP market mechanism is collusion resistant if a coalition of provers and users (\cref{def:user-collusion-resistance}) or a coalition of provers (\cref{def:prover-collusion-resistance}) results in higher joint utility and harms the utility of other parties.
\begin{definition}[User Collusion Resistance]
\label{def:user-collusion-resistance}
A ZKP market mechanism is user collusion resistant if a coalition of provers and users cannot simultaneously result in higher joint utility and harm the utility of other parties.
\end{definition}

\begin{definition}[Prover Collusion Resistance]
\label{def:prover-collusion-resistance}
A ZKP market mechanism is prover collusion resistant if a coalition of provers cannot simultaneously result in higher joint utility and harm the utility of other parties.
\end{definition}

\section{Core Mechanism of \tool}
\label{sec:mechanism-double-auction}
In this section, we analyze the game-theoretic properties of the core mechanism (as specified in~\cref{subsec:desiderata}): we prove that it satisfies Budget Balance (\cref{prop:new-wbb}), UDSIC (\cref{prop:new-duic}), and PDSIC (\cref{prop:new-dpic}).
In the interest of space, we defer some of the proofs to \cref{sec:proof}.

\parhead{Budget balance}
A ZKP market mechanism is budget-balanced if the fees collected from all users are no less than the total payments to the provers.

\begin{proposition}
\label{prop:new-wbb}
The \tool mechanism is BB.
\end{proposition}
\begin{myproof}
According to the \tool mechanism, 
each allocated task is charged $\txfee_{\capsum_\ell + 1}$, and each allocated prover is paid the unit price $\cost_{\ell  + 1}$.
Since the mechanism guarantees that $\txfee_{\capsum_\ell + 1} \geq \cost_{\ell + 1}$, it must be budget-balanced.
\end{myproof}

\parhead{User incentive compatibility}
A ZKP market mechanism is UDSIC if the mechanism is DSIC for users (\cref{def:uic}).
In~\cref{prop:new-duic}, we prove that the core mechanism of \tool is UDSIC.

\begin{proposition}
\label{prop:new-duic}
The \tool mechanism is UDSIC.
\end{proposition}

\begin{proof}
We show that for any user, bidding differently than its valuation is not profitable.
Suppose user $i$ changes its bid from $\txfee_i$ to $\txfee$. We re-label the modified user bids as $\txfee'_1 \geq \cdots \geq \txfee'_\numuser$.
Similarly, we use primed variables throughout the proof to denote variables after user $i$ changes its bid, and those without primes denote variables before that.
I.e., $\ell'$ is the number of allocated provers and $\capsum'_{\ell'}$ the allocated number of tasks after $i$'s change. 
Note that prover bids do not change by assumption, so $\cost'_j = \cost_j$ and $\capsum'_j = \capsum_j$ for all $j = 1, \cdots, \numprover$.

We say a task has \textit{rank $i$} if its fee is the $i$-th highest among all tasks. We denote the rank of user $i$'s task before and after the deviation with $R(i)$ and $R'(i)$, respectively.
We now discuss the utility of user $i$ in two cases: when the task is allocated before the change (i.e., $R(i) \in [1, \capsum_\ell]$), and when it is not.

\parhead{Case 1: User $i$ is allocated before changing the bid.}
By the assumption that user $i$'s task is allocated before deviation, the fee $\txfee_i$ is at least $\txfee_{\capsum_\ell + 1}$, i.e., $\txfee_i \ge \txfee_{\capsum_\ell + 1}$.
By the definition of our mechanism, $\cost_{\ell + 1} \le \txfee_{\capsum_\ell + 1}$ (as shown in~\cref{fig:uic-proof-1})
We now consider two sub-cases, based on whether $f < \txfee_{\capsum_\ell + 1}$ or not, i.e., whether the modified fee is less than the ``second'' price before deviating.

\begin{figure}[h]
    \centering
    \includegraphics[width = 0.6\textwidth,page=6]{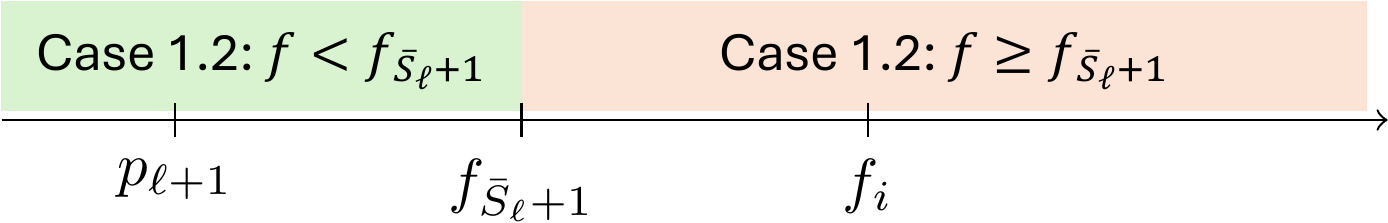}
    \caption{Two of the sub-cases examined in the proof of \cref{prop:new-duic}.}
    \label{fig:uic-proof-1}
\end{figure}

If $\txfee \ge \txfee_{\capsum_\ell + 1}$, we show that the task of user $i$ is still allocated, and the second price does not change.
Otherwise, if $\txfee < \txfee_{\capsum_\ell + 1}$, we show that the task of user $i$ is not allocated.
In both, the utility of user $i$ does not increase.

\textbf{Case 1.1: $\txfee \geq \txfee_{\capsum_\ell + 1}$.}
We now show that in this case, user $i$ is still allocated and the payment $\txfee'_{\capsum'_{\ell'} + 1}$ remains the same as $\txfee_{\capsum_{\ell} + 1}$.

Firstly, we show that $R'(i)\le\capsum_\ell$.
By assumption ($\txfee \geq \txfee_{\capsum_\ell + 1}$ and $\txfee_i \geq \txfee_{\capsum_\ell + 1}$),
tasks with a rank greater than $\capsum_\ell$ before the deviation will have the same rank after the deviation.
\Cref{fig:uicproof} shows an intuitive example.

\begin{figure}[h]
\centering
    \includegraphics[width=0.7\textwidth,page=5]{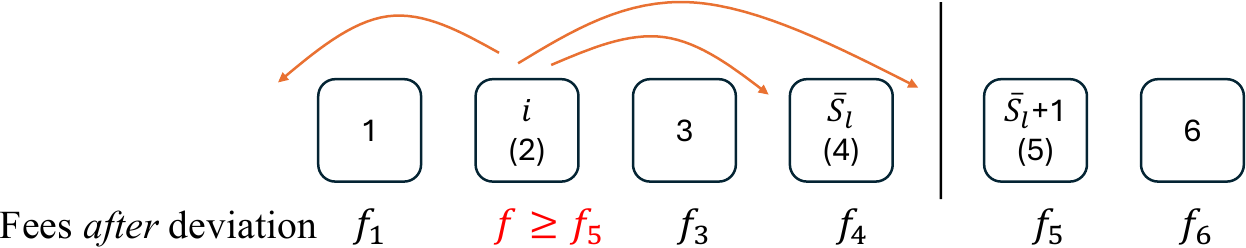}
    \caption{Example used in case 1.1 to show $R'(i) \le \capsum_\ell$. Labels show the ranks of tasks before $i$'s deviation. After the deviation, $i$'s new rank $R'(i)\le 4$, because $\txfee \geq \txfee_{\capsum_\ell + 1}$.}
    \label{fig:uicproof}
\end{figure}
In this example, $i =2$ and $\capsum_\ell = 4$.
After user $i$ changes its bid such that $\txfee \ge \txfee_{\capsum_\ell + 1}$, tasks with rank $5$ and $6$ will still have the same rank, so $R'(i)$ can only be in $[1, 4]$. In general, it follows that $R'(i) \le \capsum_\ell$.

Secondly, we prove that $\ell' = \ell$, by showing that $\ell' \geq \ell$ and $\ell' \leq \ell$.
To see that $\ell' \geq \ell$, recall that $\ell \overset{\text{def}}{=} \underset{j}{\arg\max} \{ \cost_{j + 1} \leq \txfee_{\capsum_j + 1} \}$, and $\ell' \overset{\text{def}}{=} \underset{j}{\arg\max} \{ \cost'_{j + 1} \leq \txfee'_{\capsum'_j + 1} \}$.
As argued previously, $\txfee'_{\capsum_\ell + 1} = \txfee_{\capsum_\ell + 1}$ and $R'(i) \le \capsum_\ell$ (c.f.~\cref{fig:uicproof}).
In addition, note that $\cost'_{\ell+1} = \cost_{\ell + 1}$ and $\capsum'_{\ell} = \capsum_{\ell}$ since prover bids are not changed, indicating $\txfee'_{\capsum_\ell + 1} = \txfee'_{\capsum'_\ell + 1}$.
It follows then 
$
\cost'_{\ell + 1} = \cost_{\ell + 1} \leq \txfee_{\capsum_\ell + 1} = \txfee'_{\capsum_\ell + 1} = \txfee'_{\capsum'_\ell + 1}
$. Since $\ell'$ is the maximum of all such $\ell$'s, we have $\ell' \geq \ell$.

To see $\ell' \leq \ell$, consider any $\ell^\ast > \ell$.
We have $\cost_{\ell^\ast + 1} > \txfee_{\capsum_{\ell^\ast} + 1}$ by the definition of $\ell$.
Note that $f'_{\capsum_{\ell^\ast} + 1} = f_{\capsum_{\ell^\ast} + 1}$, since $\capsum_{\ell^\ast} + 1 > \capsum_{\ell}$ by the definition of $\capsum$, and that tasks with a rank greater than $\capsum_\ell$ will have the same rank after the deviation (as argued previously).
Further recall that $\cost'_{\ell^\ast+1} = \cost_{\ell^\ast + 1}$ and $\capsum_{\ell^\ast} = \capsum'_{\ell^\ast}$ since prover bids are not changed.
It follows that
$
\cost'_{\ell^\ast + 1} = \cost_{\ell^\ast + 1} > \txfee_{\capsum_{\ell^\ast} + 1} = \txfee'_{\capsum_{\ell^\ast} + 1} = \txfee'_{\capsum'_{\ell^\ast} + 1}
$.
Since $\cost'_{\ell^\ast + 1} > \txfee'_{\capsum'_{\ell^\ast} + 1}$ for any $\ell^\ast > \ell$, $\ell$ must be at least maximum, i.e., $\ell'\leq \ell$.

In summary, $\ell' = \ell$.
We now can derive $\capsum'_{\ell'}$ the number of allocated tasks after the deviation: $\capsum'_{\ell'} = \capsum'_{\ell} = \capsum_{\ell}$.
Since $R'(i) \le \capsum_{\ell} = \capsum'_{\ell'}$ as argued previously, task $i$ remains allocated after the deviation.
The payment after deviating is $f'_{\capsum'_{\ell'} + 1} =f'_{\capsum_{\ell} + 1} = f_{\capsum_{\ell} + 1}$, as before. Thus, the utility of user $i$ remains the same.

The proof for case 1.2 (when $\txfee < \txfee_{\capsum_\ell + 1}$, user $i$'s task will not be allocated after the user changes its bid) is similar to that for case 1.1; the proof for case 2 is more involved as it incorporates more sub-cases.
\end{proof}

\parhead{Prover incentive compatibility}
Considering the prover side, recall that a ZKP market is PDSIC if the mechanism is DSIC for the provers to bid their costs honestly.
We show that the \tool mechanism is PDSIC.

\begin{proposition}
\label{prop:new-dpic}
The \tool mechanism is PDSIC.
\end{proposition}

\begin{proof}
We show that it is not profitable for any prover to bid other than its valuation.
Recall that the users' bids satisfy $\txfee_1 \geq \cdots \geq \txfee_\numuser$, and the provers' bids satisfy $\cost_1 \leq \cdots \leq \cost_\numprover$.
We assume that the prover $j$ changes its bid from $\cost_j$ to $\cost$, and suppose
after the bid change, the new user bids are $\txfee'_1 \geq \cdots \geq \txfee'_\numuser$, and the new prover bids are $\cost'_1 \leq \cdots \leq \cost'_\numprover$.
(All notations with primes denote the variables after prover $j$ changes its bid, but note that the user bids are not changed, i.e., $\txfee'_i = \txfee_i$ for all $i = 1, \cdots, \numuser$.)
Let $\ell'$ be the number of allocated provers is and $\capsum'_{\ell'}$ be the number of allocated tasks.

A prover is said to have {\em rank $j$} if its cost is the $j$-th lowest among all provers.
We denote the rank of prover $j$ before and after changing its bid with $R(j)$ and $R'(j)$.
We now discuss the utility change of prover $j$ in two cases: when the prover is allocated before the change (i.e., $R(j) \le \ell$), and when it is not.

\textbf{Case 1: Prover $j$ is allocated before changing the bid.} 
By the assumption that prover $j$ is allocated, the cost $\cost_j$ is at most $\cost_{\ell + 1}$, i.e., $\cost_j \le \cost_{\ell + 1}$.
By the mechanism we have $\cost_{\ell + 1} \le \txfee_{\capsum_\ell + 1}$.
The proof proceeds in two cases depending on whether $\cost > \cost_{\ell + 1}$, i.e., the cost $\cost$ is greater than the ``second-price'' of the provers.
If $\cost \le \cost_{\ell + 1}$, we show that prover $j$ is still allocated, and the second price does not change.
Otherwise, we prover that prover $j$ is not allocated.
In either case, prover $j$'s utility does not increase.

\textbf{Case 1.1: $\cost \leq \cost_{\ell + 1}$.}
We now show that in this case, prover $j$ will still be allocated and the payment $\capacity_j \cdot \cost'_{\ell' + 1}$ remains the same as $\capacity_j \cdot \cost_{\ell + 1}$.

First, we show that $R'(j) \le \ell$.
By assumption, we have $\cost \leq \cost_{\ell + 1}$ and $\cost_j \leq \cost_{\ell + 1}$, thus provers with rank greater than $\ell$ before prover $j$ changes its bid will have the same rank after it changes its bid.
It follows that $R'(j) \le \ell$.

Second, we show that $\ell' = \ell$, as $\ell' \geq \ell$ and $\ell' \leq \ell$.
Recall that by definition, $\ell = \underset{k}{\arg\max} \{ \cost_{k + 1} \leq \txfee_{\capsum_k + 1} \}$, and $\ell' = \underset{k}{\arg\max} \{ \cost'_{k + 1} \leq \txfee'_{\capsum'_k + 1} \}$.
To see that $\ell' \geq \ell$, note that $\capsum'_{\ell} = \capsum_\ell$ and $\cost'_{\ell + 1} = \cost_{\ell + 1}$, because the provers with rank greater than $\ell$ before prover $j$ changes its bid will have the same rank after prover $j$ changes its bid (as argued previously).
Further recall that $\txfee'_{\capsum_\ell} = \txfee_{\capsum_\ell}$ since user bids are not changed.
It follows that
$
\cost'_{\ell + 1} = \cost_{\ell + 1} \leq \txfee_{\capsum_\ell + 1} = \txfee'_{\capsum_\ell + 1} = \txfee'_{\capsum'_\ell + 1}
$.
So, $\ell' \geq \ell$ by the definition of $\ell'$.

To see $\ell' \leq \ell$, consider any $\ell^\ast > \ell$, and $\cost_{\ell^\ast + 1} > \txfee_{\capsum_{\ell^\ast} + 1}$ by the definition of $\ell$.
Because the provers with rank greater than $\ell$ before prover $j$ changes its bid will have the same rank after prover $j$ changes its bid (as argued previously), $\cost'_{\ell^\ast + 1} = \cost_{\ell^\ast + 1}$ and $\capsum'_{\ell^\ast} = \capsum_{\ell^\ast}$.
Note that $\txfee'_{\capsum_{\ell^\ast} + 1} = \txfee_{\capsum_{\ell^\ast} + 1}$ since user bids are not changed.
It follows that
$
\cost'_{\ell^\ast + 1} = \cost_{\ell^\ast + 1} > \txfee_{\capsum_{\ell^\ast} + 1} = \txfee'_{\capsum_{\ell^\ast} + 1} = \txfee'_{\capsum'_{\ell^\ast} + 1}
$.
Since $\cost'_{\ell^\ast + 1} > \txfee'_{\capsum'_{\ell^\ast} + 1}$ for any $\ell^\ast > \ell$, $\ell' \leq \ell$ by the definition of $\ell'$.
In summary, $\ell' = \ell$.

We now have the allocation after prover $j$ changes its bid: $\cost'_{\ell' + 1} = \cost'_{\ell + 1} = \cost_{\ell + 1}$.
Therefore, prover $j$ will still be allocated.
Moreover, since the payment to of prover $j$ is $\cost'_{\ell'+ 1} \cdot \capacity_j = \cost_{\ell + 1} \cdot \capacity_j$, the utility of prover $j$ will remain the same.

The proof of case 1.2 (when $\cost > \cost_{\ell + 1}$, prover $j$ will not be allocated after it changes its bid) and the proof of case 2 are more involved.
\end{proof}

\section{Implementing \texorpdfstring{\tool}{ϕ} in an ZKP Market}
\label{sec:system-level-design}
In the previous section, we have shown that \tool has desirable game-theoretic properties.
However, the core mechanism alone leaves certain security threats open.
In this section, we present system-level designs to mitigate them.
As before, missing proofs are given in \cref{sec:proof}.

\subsection{Missing or Incorrect ZKPs}
\label{subsec:missing-zkp}
It is possible that an allocated prover does not generate the required ZKPs.
To address this concern, we require each prover to deposit collateral when joining the protocol. The collateral is used to refund users whose tasks were not generated on time.
{
Specifically, suppose the market designer sets a public ``refund'' limit for user tasks that are not completed, up to a value of $\bar{\cost}$.
This effectively serves as an upper bound on user valuations: if the refund $\bar{\cost}$ seems too low to some users, they can choose not to send their task to the ZKP market.}
Then, a prover with capacity $s$ needs to maintain a deposit of at least $\bar{\cost} \cdot \capacity$ when it bids $(\capacity, \cost)$.
Note that $\cost < \bar{\cost}$.
If the prover does not finish the allocated task in time, its deposit will be confiscated (causing it to lose $\capacity \cdot \bar{\cost}$ utility) and used to refund $\bar{\cost}$ to each affected user, which is no less than the fees they paid.
{
We note that this way of setting the collateral matches industry-adopted practices such as~\cite{roy2024succinct}.}

\subsection{Misreporting Capacity}
\Cref{prop:new-dpic} proved that provers are incentivized to bid their costs honestly.
However, a prover may misreport its \textit{capacity}, as demonstrated in~\cref{eg:new-not-pic}.
\begin{example}
\label{eg:new-not-pic}
We continue with~\cref{eg:brief-example-intro}, where
when all parties bid honestly, $\ell = 1$ and prover $1$ gets $(1 - 0) \times 4 = 4$ utility.
However, when prover $1$ bids capacity $1$ instead of $4$, $\ell = 2$ and prover $1$ gets $(10 - 0) \times 1 = 10$ utility. The allocation is shown pictorially as follows.

\begin{center}
\includegraphics[width=0.9\linewidth,page=2]{figure/proofee-crop.pdf}
\end{center}

\end{example}

This example suggests that reporting a lower capacity in some cases can render more provers feasible, increasing the prover's payment.
Worst yet, fewer users are allocated in the market, and those allocated must pay more.
However, in~\cref{prop:capacity-misreport-fail}, we observe that the above attack is profitable only under specific market conditions.
The mitigation we propose is based on this observation.
\begin{proposition}
\label{prop:capacity-misreport-fail}
    When a prover $j < N$ submits a smaller capacity bid than its true capacity,
    there will always be some prover and user bids that result in the prover receiving lower utility compared to if it had bid honestly, assuming all other provers are bidding honestly.
\end{proposition}

\parhead{Mitigation: Fixing capacity}
As the profitability of the above attack depends on user bids, one mitigation is to have provers submit their capacity bids before users.
The uncertainty of user bids creates a risk for provers.
Another approach is to fix prover capacity for several rounds, to prevent provers from strategically adjusting capacities based on market conditions.
In practice, a prover's capacity is unlikely to change drastically in a short amount of time.

\subsection{User Sybil Proofness}
We show in~\cref{eg:new-user-sybil} the core mechanism alone is not user Sybil proof, i.e., provers can profit by creating fake user bids.

\begin{example}
\label{eg:new-user-sybil}
Suppose there are $8$ user tasks with values $(10,10,10,2,2,2,2,2)$, and $3$ provers with capacities $\capacity_i$ and costs $\cost_i$ $(\capacity_1, \cost_1) = (4, 0), (\capacity_2, \cost_2) = (2, 2), (\capacity_3, \cost_3) = (2, 9)$.
When all parties bid honestly, we have $\ell = 1$ and prover $1$ can make $(2 -0) \times 4 = 8$.
However, prover $1$ can create $4$ Sybil tasks with a bid of $9$.
Then we have $\ell = 2$, and prover 1 will make $\underbrace{(9 - 0) \times 4}_\text{prover reward} - \underbrace{9 \times 3}_\text{task cost} = 9$.
\begin{center}
\includegraphics[width=0.7\linewidth,page=3]{figure/proofee-crop.pdf}
\end{center}
\end{example}

Intuitively, creating high-paying tasks can sometimes render more provers feasible, increasing all provers' utility. This can also reduce user utility.
We observe that creating user Sybils benefits the attacker only under specific conditions and that the attack is possible only if the malicious prover knows all user bids.

\begin{proposition}
\label{prop:prover-sybil-user}
    When a prover $j < N$ submits Sybil user bids, there will always be some prover and user bids that result in prover $j$ receiving lower utility compared to if it had acted honestly, assuming all other provers are honest.
\end{proposition}

\parhead{Mitigation: hiding bids} Based on this observation, the above attacks can be mitigated by \textit{encrypting all user and prover bids} to hide the necessary information required by profitable Sybil attacks, i.e., making the auctions sealed-bid.

\begin{remark}
If provers are risk-averse and thus would try to avoid worst-case outcomes (as assumed by prior work, see \cref{rem:risk}), \cref{prop:prover-sybil-user} suffices to thwart the threat of a deviating prover sending Sybil user bids.
One may consider a Bayesian setting, where different solutions may be needed to ensure that such deviations lower an attacker's \emph{expected} utility (see~\cref{sec:conclusions}).
\end{remark}

\subsection{Prover Sybil proofness}
Recall that a mechanism is prover Sybil proof if no prover finds it profitable to submit bids as multiple provers.
We show that the \tool mechanism is not prover Sybil proof, as demonstrated in~\cref{eg:new-prover-sybil}.
Intuitively, prover $\ell$ may split into multiple provers and get allocated some proof generation capacity.

\begin{example}
\label{eg:new-prover-sybil}
Suppose there are $8$ user tasks, where $3$ of them have value $10$, and $5$ of them have value $3$.
Also, suppose there are $3$ provers in the ZKP market.
The capacities $\capacity_i$ and costs $\cost_i$ of the provers are $(\capacity_1, \cost_1) = (4, 0), (\capacity_2, \cost_2) = (2, 2), (\capacity_3, \cost_3) = (2, 9)$.
When all parties are bidding honestly, we have $\ell = 1$, and prover $2$ will get zero utility.
Instead, prover $2$ can split into two provers, one with bid $(1, 2)$ and the other with bid $(1, 3)$.
Then $\ell$ will increase to $2$, and the split prover with bid $(1, 2)$ will be allocated and get $1$ utility.

\begin{center}
\includegraphics[width=0.7\linewidth,page=4]{figure/proofee-crop.pdf}
\end{center}
\aviv{Could add a caption and reference the image.}
\end{example}

We analyze prover Sybil attacks and show in \cref{prop:sybil-good} that profitable attacks contribute to social welfare and thus do not harm users.
\begin{proposition}
\label{prop:sybil-good}
    Suppose a prover $j$ with capacity $\capacity_j$ splits into $p$ Sybil provers, where the $k$-th split has cost bid $\cost_{j, k}$ and capacity bid $\capacity_{j, k}$, with $\sum_{k}\capacity_{j, k} = \capacity_j$.
    Then if the Sybil attack is profitable, the social welfare will increase.
\end{proposition}

\subsection{Collusion}
A mechanism is vulnerable to collusion if a coalition of provers and/or users finds it profitable to coordinate and bid strategically.
In this paper, we focus on two types of collusion in a ZKP market; one is a prover colluding with a user, and the other is a set of provers colluding with each other.
{We note that potential collusion among market agents is typically not discussed in related work~\cite{gong2024v3rified,huang2002design}.
Here, we consider the possibility of collusion and show that it is not much of a threat with proper implementation of the mechanism.
Specifically, we assume a coalition of colluding participants can freely communicate; i.e., colluders have private communication channels with each other, while honest actors not in the coalition are independent and do not communicate with any other party.}

In case a prover colludes with a user, we show in \cref{prop:collusion-prover-user} that this benefits the attacking coalition only under specific conditions.
\begin{proposition}
\label{prop:collusion-prover-user}
    When a prover and a user collude by deviating from their honest bids, there will always be other prover and user bids that cause the coalition to receive lower joint utility, assuming all other parties are honest.
\end{proposition}
In the proof, we show in a case-by-case manner that all profitable deviations require the coalition to know the bids of the other parties.

Following this result, the previous mitigation (hiding bids) also mitigates such collusion by creating uncertainty about profitability.
I.e., a prover-user coalition cannot be sure that they will profit by bidding strategically.

When a set of provers form a coalition and bid strategically, if the provers in the coalition have the same cost then this is equivalent to a single prover splitting into prover Sybils.
As in~\cref{prop:sybil-good}, such coalitions do not harm the ZKP market.
On the other hand, we show in~\cref{prop:collusion-two-prover} that two provers colluding can only benefit the provers in the coalition when the other parties have specific bids, which again is mitigated by encrypting user and prover bids.
\fanz{added the last bit about mitigation}

\begin{proposition}
\label{prop:collusion-two-prover}
    When two provers with different costs collude by changing their bids, there will always be other prover and user bids that result in the coalition receiving lower joint utility compared to if they had bid honestly, assuming all other parties are bidding honestly.
\end{proposition}

\section{Conclusions and Future Work}
\label{sec:conclusions}
We have presented \tool, a mechanism designed to address the demands of ZKP markets, where proofs are generated in batches.
\tool consists of a core auction mechanism that matches low-cost provers with high-value user tasks. We formally analyze the mechanism and show that it ensures incentive compatibility for both users and provers and that it is budget-balanced.
We also introduce system-level designs such as slashing, fixing prover capacity bids over multiple auction rounds, and encrypting bids, to address security challenges that are not handled by the core mechanism alone.
We analyze how our measures mitigate security concerns including untruthful capacity bids, Sybil attacks, and collusion.

\parhead{Future work.}
Following the literature~\cite{chung2023foundations,gafni2024barriers,bahrani2024transaction,roughgarden2021transaction}, we consider a non-Bayesian setting.
One possible direction for future work is to extend the model to a Bayesian one where actors may deviate to increase \emph{expected} payoffs.
For example, provers can send Sybil user bids to ``simulate'' a reserve price for users, thus potentially increasing revenue.
Analyzing this manipulation could be interesting from the perspective of the adversaries, who should surmount several practical difficulties.
First, attackers should have robust estimates for the distribution of user valuations and the distribution of the number of users.
Second, in a permissionless setting, price gouging may not be profitable in the long run: it can allow new actors to offer lower prices while still making a profit~\cite{huberman2021monopoly}.
Moreover, long-term dynamics may result in lower-paying transactions accumulating until allocating the pent-up demand becomes profitable for provers~\cite{nisan2023serial}.

\section*{Acknowledgments}
The authors wish to thank Matter Labs for research discussions and support with several aspects of this work.
Wenhao Wang was supported in part by the ZK Fellowship from Matter Labs.

\printbibliography
\newpage
\section*{Appendix}
\appendix

\section{Proofs}
\label{sec:proof}
\subsection{Full Proof of~\cref{prop:new-duic}}
\begin{myproof}
    We need to show that it is not profitable for any user to bid other than its valuation.
    Recall that the users' bids satisfy $\txfee_1 \geq \cdots \geq \txfee_\numuser$, and the provers' bids satisfy $\cost_1 \leq \cdots \leq \cost_\numprover$.
    Suppose that $\txfee_i$ is user $i$'s valuation of its task.
    We assume that user $i$ changes its bid from $\txfee_i$ to $\txfee$.
    After the bid change, the new user bids are $\txfee'_1 \geq \cdots \geq \txfee'_\numuser$, and the new prover bids are $\cost'_1 \leq \cdots \leq \cost'_\numprover$.
    The allocated number of provers are $\ell'$ and the allocated number of tasks are $\capsum'_{\ell'}$
    Since no prover's bid has changed, we have $\cost'_j = \cost_j$ and $\capsum'_j = \capsum_j$ for all $j = 1, \cdots, \numprover$.
    We discuss the utility change of user $i$ depending on whether or not it is allocated, i.e., whether or not $i \in [1, \capsum_\ell]$.
    A task is said to have \textit{rank $i$} if its fee is the $i$-th highest among all tasks.
    Let $R(i)$ and $R'(i)$ be the rank of user $i$'s task before and after changing her bid.
    In each case, a figure is used to demonstrate the relation between variables, and there could be more possibilities.
    \begin{itemize}[leftmargin=*]
    \item If $i \in [1, \capsum_\ell]$, i.e., task $i$ is allocated, then we consider whether $\txfee < \txfee_{\capsum_\ell + 1}$.
    
    \begin{itemize}[leftmargin=*]
    \item If $\txfee \geq \txfee_{\capsum_\ell + 1}$, the concern is that user $i$ can somehow decrease its payment while its task is still allocated after the bid change.

    \begin{center}
    \begin{tikzpicture}
    \draw[->] (-3,0) -- (5,0) node[right] {};
    \draw (-2,0.1) -- (-2,-0.1) node[below] {$\cost_{\ell + 1}$};
    \draw (0,0.1) -- (0,-0.1) node[below] {$\txfee_{\capsum_\ell + 1}$};
    \draw (1,0.1) -- (1,-0.1) node[below] {$\txfee$};
    \draw (3,0.1) -- (3,-0.1) node[below] {$\txfee_i$};
    \end{tikzpicture}
    \end{center}
    
    We show that in this case task $i$ is still allocated, while the payment $\txfee'_{\capsum'_{\ell'} + 1}$ remains the same.

    Firstly, we show that $R'(i) \le \capsum_\ell$ after user $i$ changes its bid to $\txfee$.
    This is because $\txfee \geq \txfee_{\capsum_\ell + 1}$ and $\txfee_i \geq \txfee_{\capsum_\ell + 1}$, which indicates that tasks with rank greater than $\capsum_\ell$ before user $i$ changes its bid will have the same rank after user $i$ changes its bid.
    It follows that $R'(i) \le \capsum_\ell$

    Secondly, we show that $\ell' = \ell$.
    We show this with $\ell' \geq \ell$ and $\ell' \leq \ell$.
    To see that $\ell' \geq \ell$, recall that $\ell = \underset{j}{\arg\max} \{ \cost_{j + 1} \leq \txfee_{\capsum_j + 1} \}$, and $\ell' = \underset{j}{\arg\max} \{ \cost'_{j + 1} \leq \txfee'_{\capsum'_j + 1} \}$.
    Note that $\txfee'_{\capsum_\ell + 1} = \txfee_{\capsum_\ell + 1}$, since $R'(i) \le\capsum_\ell$.
    Further note that $\cost'_{\ell+1} = \cost_{\ell + 1}$ and $\capsum'_{\ell} = \capsum_{\ell}$.
    It follows that 
    $$
    \cost'_{\ell + 1} = \cost_{\ell + 1} \leq \txfee_{\capsum_\ell + 1} = \txfee'_{\capsum_\ell + 1} = \txfee'_{\capsum'_\ell + 1}.
    $$
    So $\ell' \geq \ell$ by the definition of $\ell'$.
    To see $\ell' \leq \ell$, for any $\ell^\ast > \ell$ we have $\cost_{\ell^\ast + 1} > \txfee_{\capsum_{\ell^\ast} + 1}$ by the definition of $\ell$.
    Note that $f'_{\capsum_{\ell^\ast} + 1} = f_{\capsum_{\ell^\ast} + 1}$, since $\capsum_{\ell} < \capsum_{\ell^\ast} + 1$ by the definition of $\capsum$, and tasks with rank greater than $\capsum_\ell$ before user $i$ changes its bid will have the same rank after user $i$ changes its bid (as argued previously).
    Further recall that $\cost'_{\ell^\ast+1} = \cost_{\ell^\ast + 1}$ and $\capsum_{\ell^\ast} = \capsum'_{\ell^\ast}$.
    It follows that for any $\ell^\ast > \ell$,
    $$
    \cost'_{\ell^\ast + 1} = \cost_{\ell^\ast + 1} > \txfee_{\capsum_{\ell^\ast} + 1} = \txfee'_{\capsum_{\ell^\ast} + 1} = \txfee'_{\capsum'_{\ell^\ast} + 1}.
    $$
    Therefore, $\ell'\leq \ell$ by the definition of $\ell'$.
    Above all, $\ell' = \ell$.

    We have $\capsum'_{\ell'} = \capsum'_{\ell} = \capsum_{\ell}$.
    Since the rank of task $i$ is no more than $\capsum_{\ell}$, task $i$ will still be allocated after user $i$ changes its bid to $\txfee$.
    Also, note that the payment of this task is $f'_{\capsum'_{\ell'} + 1} =f'_{\capsum_{\ell} + 1} = f_{\capsum_{\ell} + 1}$, we have the utility of user $i$ will remain the same after user $i$ changes its bid.

    \item If $\txfee < \txfee_{\capsum_\ell + 1}$, the concern is also that user $i$ can somehow decrease its payment while its task is still allocated after it changes its bid to $\txfee$.

    \begin{center}
    \begin{tikzpicture}
    \draw[->] (-3,0) -- (5,0) node[right] {};
    \draw (-2,0.1) -- (-2,-0.1) node[below] {$\cost_{\ell + 1}$};
    \draw (-1,0.1) -- (-1,-0.1) node[below] {$\txfee$};
    \draw (0,0.1) -- (0,-0.1) node[below] {$\txfee_{\capsum_\ell + 1}$};
    
    \draw (3,0.1) -- (3,-0.1) node[below] {$\txfee_i$};
    \end{tikzpicture}
    \end{center}
    
    We show that in this case task $i$ is not allocated, so user $i$ cannot increase its utility after it changes its bid.

    Firstly, we show that $R'(i) > \capsum_\ell$.
    This is because $\txfee < \txfee_{\capsum_\ell + 1}$ and $\txfee_i \geq \txfee_{\capsum_\ell + 1}$, which indicates that there will be at least $\capsum_\ell$ tasks with bids greater than $\txfee$ after user $i$ changes its bid.
    It follows that $R'(i) > \capsum_\ell$.

    Secondly, we show that $\ell' \leq \ell$.
    To see this, for all $\ell^\ast > \ell$ we have $\cost_{\ell^\ast + 1} > \txfee_{\capsum_{\ell^\ast} + 1}$ by the definition of $\ell$.
    Note that $f'_{\capsum_{\ell^\ast} + 1} \leq f_{\capsum_{\ell^\ast} + 1}$, since $\txfee < \txfee_i$ and when there is no increasing element in a sequence, element in the new sequence will be no more than the element with the same rank in the original sequence.
    Further recall that $\cost'_{\ell^\ast+1} = \cost_{\ell^\ast + 1}$ and $\capsum_{\ell^\ast} = \capsum'_{\ell^\ast}$.
    It follows that
    $$
    \cost'_{\ell^\ast + 1} = \cost_{\ell^\ast + 1} > \txfee_{\capsum_{\ell^\ast} + 1} \geq \txfee'_{\capsum_{\ell^\ast} + 1} = \txfee'_{\capsum'_{\ell^\ast} + 1}.
    $$
    Therefore, $\ell'\leq \ell$ by the definition of $\ell'$.

    We have $\capsum'_{\ell'} \leq \capsum'_{\ell} = \capsum_{\ell}$.
    Since $R'(i) > \capsum_{\ell}$, task $i$ will not be allocated after user $i$ changes its bid to $\txfee$.
    Therefore, the utility of user $i$ will become $0$ after user $i$ changes its bid.
    \end{itemize}
    
    \item If $i \in [\capsum_\ell + 1, \numuser]$, i.e., task $i$ is not allocated, we consider the cases depending on the value of $i$.
    We have there exist some $\tilde{\ell} \geq \ell$ such that $i \in [\capsum_{\tilde{\ell}} + 1, \capsum_{\tilde{\ell} + 1}]$ for some $\tilde{\ell} \geq \ell$, or $i > \sum_{j = 1}^{\numprover} \capacity_j$.
    \begin{itemize}[leftmargin=*]
    \item If $\txfee < \txfee_{\capsum_{\tilde{\ell}}}$, we argue that task $i$ will not be included.

    \begin{center}
    \begin{tikzpicture}
    \draw[->] (-4,0) -- (5,0) node[right] {};
    \draw (-3,0.1) -- (-3,-0.1) node[below] {$\txfee_{\capsum_{\tilde{\ell} + 1}}$};
    \draw (-2,0.1) -- (-2,-0.1) node[below] {$\txfee_{i}$};
    \draw (-1,0.1) -- (-1,-0.1) node[below] {$\txfee_{\capsum_{\tilde{\ell}} +1}$};
    \draw (1,0.1) -- (1,-0.1) node[below] {$\txfee$};
    \draw (2,0.1) -- (2,-0.1) node[below] {$\txfee_{\capsum_{\tilde{\ell}}}$};
    
    \draw (3,0.1) -- (3,-0.1) node[below] {$\cost_{\ell + 1}$};
    \end{tikzpicture}
    \end{center}

    Firstly, we show that $R'(i) > \capsum_{\tilde{\ell}}$.
    This is because $\txfee < \txfee_{\capsum_{\tilde{\ell}}}$, which indicates that there are at least $\capsum_{\tilde{\ell}}$ tasks with bids greater than $\txfee$ after user $i$ changes its bid.
    It follows that $R'(i) > \capsum_{\tilde{\ell}}$.

    Secondly, we show that $\ell' \leq \tilde{\ell}$.
    To see this, for any $\ell^\ast > \tilde{\ell} \geq \ell$ we have $\cost_{\ell^\ast + 1} > \txfee_{\capsum_{\ell^\ast}+ 1}$ by the definition of $\ell$.
    Then we show that $\txfee_{\capsum_{\ell' + 1}}\geq \txfee'_{\capsum_{\ell' + 1}}$.
    If $\txfee \geq \txfee_i$, since $\capsum_{\ell' + 1} \geq i$ we have $\txfee_{\capsum_{\ell' + 1}} = \txfee'_{\capsum_{\ell' + 1}}$; if $\txfee < \txfee_i$, we have $\txfee_{\capsum_{\ell' + 1}}\geq \txfee'_{\capsum_{\ell' + 1}}$, since when there is no increasing element in a sequence, element in the the new sequence will be no more than the element with the same rank in the original sequence.
    It follows that $\txfee_{\capsum_{\ell' + 1}}\geq \txfee'_{\capsum_{\ell' + 1}}$.
    Further recall that $\cost'_{\ell^\ast + 1} = \cost_{\ell^\ast + 1}$ and $\capsum'_{\ell^\ast} = \capsum_{\ell^\ast}$.
    Therefore, for any $\ell^\ast > \tilde{\ell}$,
    $$
    \cost'_{\ell^\ast + 1} = \cost_{\ell^\ast + 1} > \txfee_{\capsum_{\ell^\ast} + 1} \geq \txfee'_{\capsum_{\ell^\ast} + 1} = \txfee'_{\capsum'_{\ell^\ast} + 1}.
    $$
    So $\ell' \leq \tilde{\ell}$ by the definition of $\ell'$.
    We have $\capsum'_{\ell'} \leq \capsum'_{\tilde{\ell}} = \capsum_{\tilde{\ell}}$.
    Since $R'(i) > \capsum_{\tilde{\ell}}$, task $i$ will not be allocated after user $i$ changes its bid to $f$.
    Therefore, the utility of user $i$ will remain $0$ after user $i$ changes its bid.
    
    \item If $\txfee \geq \txfee_{\capsum_{\tilde{\ell}}}$, the concern is that user $i$ can somehow get its task allocated will paying less than $\txfee_i$.

    \begin{center}
    \begin{tikzpicture}
    \draw[->] (-4,0) -- (5,0) node[right] {};
    \draw (-3,0.1) -- (-3,-0.1) node[below] {$\txfee_{\capsum_{\tilde{\ell} + 1}}$};
    \draw (-2,0.1) -- (-2,-0.1) node[below] {$\txfee_{i}$};
    \draw (-1,0.1) -- (-1,-0.1) node[below] {$\txfee_{\capsum_{\tilde{\ell}} +1}$};
    \draw (0,0.1) -- (0,-0.1) node[below] {$\txfee_{\capsum_{\tilde{\ell}}}$};
    \draw (2,0.1) -- (2,-0.1) node[below] {$\txfee$};
    \draw (3,0.1) -- (3,-0.1) node[below] {$\cost_{\ell + 1}$};
    \end{tikzpicture}
    \end{center}
    
    We show that in this case either task $i$ is not allocated or task $i$ is charged at least $\txfee_i$.
    
    Firstly, we show that $\ell' \leq \tilde{\ell}$.
    To see this, for any $\ell^\ast > \tilde{\ell} \geq \ell$ we have $\cost_{\ell^\ast + 1} > \txfee_{\capsum_{\ell^\ast}+ 1}$ by the definition of $\ell$.
    Then we show that $\txfee_{\capsum_{\ell' + 1}}\geq \txfee'_{\capsum_{\ell' + 1}}$.
    If $\txfee \geq \txfee_i$, since $\capsum_{\ell' + 1} \geq i$ we have $\txfee_{\capsum_{\ell' + 1}} = \txfee'_{\capsum_{\ell' + 1}}$; if $\txfee < \txfee_i$, we have $\txfee_{\capsum_{\ell' + 1}}\geq \txfee'_{\capsum_{\ell' + 1}}$, since when there is no increasing element in a sequence, element in the the new sequence will be no more than the element with the same rank in the original sequence.
    It follows that $\txfee_{\capsum_{\ell' + 1}}\geq \txfee'_{\capsum_{\ell' + 1}}$.
    Further recall that $\cost'_{\ell^\ast + 1} = \cost{\ell^\ast + 1}$ and $\capsum'_{\ell^\ast} = \capsum_{\ell^\ast}$.
    Therefore, for any $\ell^\ast > \tilde{\ell}$,
    $$
    \cost'_{\ell^\ast + 1} = \cost_{\ell^\ast + 1} > \txfee_{\capsum_{\ell^\ast} + 1} \geq \txfee'_{\capsum_{\ell^\ast} + 1} = \txfee'_{\capsum'_{\ell^\ast} + 1}.
    $$
    So $\ell' \leq \tilde{\ell}$ by the definition of $\ell'$.
    
    We have $\capsum'_{\ell'} \leq \capsum'_{\tilde{\ell}} = \capsum_{\tilde{\ell}}$.
    If task $i$ is allocated, it is charged $\txfee'_{\capsum'_{\ell'}}\geq \txfee'_{\capsum'_{\tilde{\ell}}}$.
    Since $\txfee > \txfee_i$, and when there is no decreasing element in a sequence, the element in the new sequence will be no less than the element with the same rank in the original sequence.
    So, $\txfee'_{\capsum'_{\tilde{\ell}}} \geq \txfee_{\capsum'_{\tilde{\ell}}} = \txfee_{\capsum_{\tilde{\ell}}} \geq \txfee_i.$
    It follows that the payment of task $i$ is at least $\txfee_i$ if task $i$ is allocated.
    Therefore, user $i$ cannot get positive utility if it changes its bid to $\txfee$.
    \end{itemize}
    \end{itemize}

    Above all, the mechanism is UDSIC.
    \wenhao{TODO: give a pass, look at notations}
\end{myproof}

\subsection{Full Proof of~\cref{prop:new-dpic}}
\begin{myproof}
    We prove it is unprofitable for any prover to bid other than its valuation.
    Recall that the users' bids satisfy $\txfee_1 \geq \cdots \geq \txfee_\numuser$, and the provers' bids satisfy $\cost_1 \leq \cdots \leq \cost_\numprover$.
    Suppose that $\cost_j$ is prover $n$'s valuation of its proof generation cost per unit capacity.
    We assume that the prover $j$ changes its bid from $\cost_j$ to $\cost$.
    After the bid change, the new user bids are $\txfee'_1 \geq \cdots \geq \txfee'_\numuser$, and the new prover bids are $\cost'_1 \leq \cdots \leq \cost'_\numprover$.
    The allocated number of provers is $\ell'$ and the allocated number of tasks is $\capsum'_{\ell'}$.
    Since no user's bid has changed, we have $\txfee'_i = \txfee_i$ for all $i = 1, \cdots, \numuser$.
    We discuss the utility change of prover $j$ depending on whether or not it is allocated, i.e., whether or not $j \in [1, \ell]$.
    A prover is said to have {\em rank $j$} if its cost is the $j$-th lowest among all provers.
    We denote the rank of prover $j$ before and after changing its bid with $R(j)$ and $R'(j)$.
    \wenhao{Added the following:}
    Each case's figure demonstrates the relation between variables, but there could be more possibilities.

    \begin{itemize}[leftmargin=*]
    \item If prover $j$ is allocated ($j \in [1, \ell]$), consider possible relations between $\cost, \cost_{\ell + 1}$.
    \begin{itemize}[leftmargin=*]
    \item If $\cost \leq \cost_{\ell + 1}$, the concern is that prover $j$ can increase its revenue while remaining allocated.

    \begin{center}
    \begin{tikzpicture}
    \draw[->] (-3,0) -- (5,0) node[right] {};
    \draw (-2,0.1) -- (-2,-0.1) node[below] {$\cost_j$};
    \draw (0,0.1) -- (0,-0.1) node[below] {$\cost$};
    
    \draw (3,0.1) -- (3,-0.1) node[below] {$\cost_{\ell + 1}$};
    \end{tikzpicture}
    \end{center}
    \aviv{The line figure above shows that $\cost > \cost_j$, while the proof below is general}
    \wenhao{I suppose it is okay}
    
    We show that in this case prover $j$ remains allocated, while the payment $\cost'_{\ell' + 1} \cdot \capacity_j$ stays the same, so changing the bid to $\cost$ is not profitable.
    
    First, we show that $R'(j) \le \ell$.
    This is because $\cost \leq \cost_{\ell + 1}$ and $\cost_j \leq \cost_{\ell + 1}$, which indicates that provers with rank greater than $\ell$ have the same rank before and after prover $j$ changes its bid.
    Thus, $R'(j) \le \ell$.

    Second, we prove that $\ell' = \ell$ by showing that both $\ell' \geq \ell$ and $\ell' \leq \ell$.
    Recall that $\ell = \underset{k}{\arg\max} \{ \cost_{k + 1} \leq \txfee_{\capsum_k + 1} \}$, and $\ell' = \underset{k}{\arg\max} \{ \cost'_{k + 1} \leq \txfee'_{\capsum'_k + 1} \}$.
    To see that $\ell' \geq \ell$, note that $\capsum'_{\ell} = \capsum_\ell$ and $\cost'_{\ell + 1} = \cost_{\ell + 1}$, since the provers with rank greater than $\ell$ before prover $j$ changes its bid will have the same rank after prover $j$ changes its bid.
    Further recall that $\txfee'_{\capsum_\ell} = \txfee_{\capsum_\ell}$.
    It follows that
    $$
    \cost'_{\ell + 1} = \cost_{\ell + 1} \leq \txfee_{\capsum_\ell + 1} = \txfee'_{\capsum_\ell + 1} = \txfee'_{\capsum'_\ell + 1}.
    $$
    So $\ell' \geq \ell$ by the definition of $\ell'$.
    To see $\ell' \leq \ell$, we have $\cost_{\ell^\ast + 1} > \txfee_{\capsum_{\ell^\ast} + 1}$ for all $\ell^\ast > \ell$ by the definition of $\ell$.
    Note that $\cost'_{\ell^\ast + 1} = \cost_{\ell^\ast + 1}$ and $\capsum'_{\ell^\ast} = \capsum_{\ell^\ast}$ since the provers with rank greater than $\ell$ before prover $j$ changes its bid will have the same rank after prover $j$ changes its bid.
    Also recall that $\txfee'_{\capsum_{\ell^\ast} + 1} = \txfee_{\capsum_{\ell^\ast} + 1}$.
    It follows that:
    $
    \cost'_{\ell^\ast + 1} = \cost_{\ell^\ast + 1} > \txfee_{\capsum_{\ell^\ast} + 1} = \txfee'_{\capsum_{\ell^\ast} + 1} = \txfee'_{\capsum'_{\ell^\ast} + 1}.
    $
    Therefore, $\ell' \leq \ell$ by the definition of $\ell'$.
    Above all, $\ell' = \ell$.

    We have $\cost'_{\ell' + 1} = \cost'_{\ell + 1} = \cost_{\ell + 1}$ as argued previously.
    Therefore, prover $j$ will be allocated and its payment is $\cost'_{\ell + 1} \cdot \capacity_j =\cost_{\ell + 1} \cdot \capacity_j $.
    Therefore, prover $j$ will not increase its utility after it changes its bid.
    
    \item If $\cost > \cost_{\ell + 1}$, the concern is that prover $j$ can increase $\ell$ and get allocated.
    \done\aviv{The text above says $\cost_j > \cost_{\ell + 1}$, but the line figure below shows that $\cost_j < \cost_{\ell + 1}$. If I understand correctly, this should be ``If $\cost > \cost_{\ell + 1}$'' (because this is the negation of the previous case). Also, in the current case, isn't the concern different? The prover was allocated before the manipulation, so the concern should be that the prover will manipulate its bid to still be allocated but get more revenue.}

    \begin{center}
    \begin{tikzpicture}
    \draw[->] (-3,0) -- (5,0) node[right] {};
    \draw (-2,0.1) -- (-2,-0.1) node[below] {$\cost_j$};
    \draw (0,0.1) -- (0,-0.1) node[below] {$\cost_{\ell + 1}$};
    
    \draw (3,0.1) -- (3,-0.1) node[below] {$\cost$};
    \end{tikzpicture}
    \end{center}
    
    We prove by contradiction that in this case prover $j$ is not allocated, thus getting zero revenue.
    Before that, we show $R'(j) > \ell$.
    This is because $\cost > \cost_{\ell + 1}$ and $\cost_j \leq \cost_{\ell + 1}$, which indicates that there will be at least $\ell$ provers with bids lower than $\cost$ after prover $j$ changes its bid.
    It follows that the rank of prover $j$ is more than $\ell$ after prover $j$ changes its bid.
    Suppose by contradiction that prover $j$ is allocated, then $R'(j) \le \ell'$, so $\ell' > \ell$.
    By the definition of $\ell'$, $\cost'_{\ell' + 1} \leq \txfee'_{\capsum'_{\ell'} + 1}$.
    Since $R'(j) \le \ell'$, then $\capsum'_{\ell'} = \capsum_{\ell'}$ and $\cost'_{\ell' + 1} = \cost_{\ell' + 1}$.
    Recall that $\txfee'_{\capsum_{\ell'} + 1} = \txfee_{\capsum_{\ell'} + 1}$.
    It follows that:
    $
    \cost_{\ell' + 1} = \cost'_{\ell' + 1} \leq \txfee'_{\capsum'_{\ell'} + 1} = \txfee_{\capsum_{\ell'} + 1}.
    $
    However, $\cost_{\ell' + 1} > \txfee_{\capsum_{\ell'} + 1}$ by the definition of $\ell$ since $\ell' > \ell$, leading to a contradiction.
    Therefore, prover $j$ is not allocated after changing its bid to $\cost$.
    Thus, prover $j$ receives zero utility.
    \end{itemize}

    \item If $j \in [\ell + 1, \numprover]$, we consider the cases depending on whether or not $\cost \geq \cost_{\ell + 1}$.
    \begin{itemize}[leftmargin=*]
    \item If $\cost < \cost_{\ell + 1}$, the concern is that prover $j$ can somehow get allocated while being paid more than $\cost_j$.

    \begin{center}
    \begin{tikzpicture}
    \draw[->] (-3,0) -- (5,0) node[right] {};
    \draw (-2,0.1) -- (-2,-0.1) node[below] {$\cost$};
    \draw (0,0.1) -- (0,-0.1) node[below] {$\cost_{\ell + 1}$};
    
    \draw (3,0.1) -- (3,-0.1) node[below] {$\cost_j$};
    \end{tikzpicture}
    \end{center}
    
    We show that in this case, prover $j$ will either not be allocated or be paid at most $\cost_j$ while being allocated.

    First, we show that $R'(j) \le \ell + 1$.
    This is because $\cost < \cost_{\ell + 1}$ and $\cost_j \geq \cost_{\ell + 1}$, which indicates that provers with rank strictly lower than $\ell +1$ before prover $j$ changes its bid will have rank lower than $\ell + 1$ after prover $j$ changes its bid, and that the prover $\ell + 1$ before prover $j$ changes its bid will have rank $\ell + 2$ after prover $j$ changes its bid.
    It follows that $R'(j) \le \ell + 1$.

    Secondly, we show that $\ell' < j$.
    By the definition of $\ell$, we have $\cost_{\ell^\ast +1} > \txfee_{\capsum_{\ell^\ast} + 1}$ for all $\ell^\ast \geq j > \ell$.
    Note that $\cost'_{\ell^\ast + 1} = \cost_{\ell^\ast + 1}$ for all $\ell^\ast \geq j$ since prover $j$ lowers its bid.
    Also recall that $\txfee'_{\capsum_{\ell^\ast} + 1} = \txfee_{\capsum_{\ell^\ast} + 1}$.
    It follows that for all $\ell^\ast \geq j$
    $$
    \cost'_{\ell^\ast + 1} = \cost_{\ell^\ast + 1} > \txfee_{\capsum_{\ell^\ast} + 1} = \txfee'_{\capsum_{\ell^\ast} + 1} = \txfee'_{\capsum'_{\ell^\ast} + 1}.
    $$
    Therefore, $\ell' < j$ by the definition of $\ell'$.

    Note that after prover $j$ changes its bid to $\cost$, it is either allocated or not allocated.
    If prover $j$ is not allocated, it gets zero utility.
    If prover $j$ is allocated, we have $\cost'_{\ell' + 1} \leq \cost'_{j} \leq \cost_j$.
    The last inequality holds because prover $j$ lowers its bid, and when there is no increasing element in a sequence, any element in the new sequence is at most equal to the element with the same rank in the original sequence.
    Therefore, prover $j$ at most earns $\cost_j \capacity_j$ and gets zero utility.
    In both cases, prover $j$ will not receive above zero utility.

    \item If $\cost \geq \cost_{\ell + 1}$, the fear is that prover $j$ can be allocated and get positive utility.

    \begin{center}
    \begin{tikzpicture}
    \draw[->] (-3,0) -- (5,0) node[right] {};
    \draw (-2,0.1) -- (-2,-0.1) node[below] {$\cost_{\ell +1}$};
    \draw (0,0.1) -- (0,-0.1) node[below] {$\cost$};
    
    \draw (3,0.1) -- (3,-0.1) node[below] {$\cost_j$};
    \end{tikzpicture}
    \end{center}
    
    We show that in this case prover $j$ cannot be allocated.

    Firstly, we show that $R'(j) \ge \ell + 1$.
    This is because $\cost \geq \cost_{\ell + 1}$ and $\cost_j \geq \cost_{\ell + 1}$, which indicates that provers with rank lower than $\ell +1$ before prover $j$ changes its bid will have rank lower than $\ell + 1$ after prover $j$ changes its bid.
    It follows that $R'(j) \ge \ell + 1$.

    Second, we show that $\ell'\geq \ell$.
    By the definition of $\ell$, we have $\cost_{\ell^\ast +1} > \txfee_{\capsum_{\ell^\ast} + 1}$ for any $\ell^\ast > \ell$.
    Note that $\cost'_{\ell + 1} \geq \cost_{\ell + 1}$, $\txfee'_{\capsum_{\ell^\ast} + 1} = \txfee_{\capsum_{\ell^\ast} + 1}$.
    Thus, for any $\ell^\ast > \ell$
    $$
    \cost'_{\ell^\ast + 1} \geq \cost'_{\ell + 1} \geq \cost_{\ell + 1} > \txfee_{\capsum_{\ell} + 1} \geq \txfee_{\capsum_{\ell^\ast} + 1} = \txfee'_{\capsum'_{\ell^\ast} + 1}.
    $$
    Therefore, $\ell' \geq \ell$ by the definition of $\ell'$.
    Since $R'(j) > \ell$, it is not allocated.
    Therefore, the utility of prover $j$ remains zero after changing its bid.
    \end{itemize}
    \end{itemize}

    Above all, the mechanism is PDSIC
\end{myproof}

\subsection{Proof of \cref{prop:capacity-misreport-fail}}
\begin{myproof}
    We discuss the utility of prover $j < N$. Suppose there are $\capsum_N$ tasks all with fee $\txfee = \cost_{j + 1}$, and $\cost_{j + 1} > \cost_j$.
    When prover $j$ is bidding honestly, $\ell = j$ and prover $j$ would receive $\capacity_j \cdot (\cost_{ j + 1} - \cost_j)$ utility.
    If prover $j$ bids a lower capacity $\capacity'_j$, still we have $\ell = j$ and prover $j$ would receive $\capacity'_j \cdot (\cost_{j + 1} - \cost_j)$ utility, which is less than the utility when the prover bids honestly.
\end{myproof}

\subsection{Proof of \cref{prop:prover-sybil-user}}
\begin{myproof}
    We discuss the utility of prover $j< N$.
    Suppose there are $\capsum_{N}$ tasks, all with fee $\txfee = \cost_{j + 1}$, and $\cost_{j + 1} > \cost_j$.
    When prover $j$ does not create Sybils, $\ell = j$ and prover $j$ receives $\capacity_j \cdot (\cost_{j + 1} - \cost_j)$ utility.
    If prover $j$ submits Sybil user bids, then prover $j$ is always allocated, and will receive the same payment.
    If any of the Sybil bids are allocated, prover $j$ will lose at least $\txfee$ utility.
    Therefore, in this case, prover $j$ does not profit by submitting Sybil user bids.

    \ignore{We further analyze the conditions where such attacks are possible.
    Assume that some prover $j$ finds it profitable to generate fake user tasks.
    Suppose after submitting fake user tasks $\ell$ increases to $\ell' > \max\{\ell + 1, j\}$.
    Suppose that $k = \underset{i}{\arg\min}\{ \txfee_i \geq \cost_{\ell'} \}$.
    Then the best strategy is to add $\capsum_{\ell'} - k + 1$ Sybil tasks at the fee of $\cost_{\ell' + 1}$.
    This is profitable if
\begin{equation}
\begin{cases}
    (\cost_{\ell' + 1} - \cost_{\ell + 1}) \cdot \capacity_j > \cost_{\ell' + 1} \cdot (\capsum_{\ell'} - k) + \cost_j \cdot \capacity_j, &j \leq \ell\\
    (\cost_{\ell' + 1} - \cost_{\ell + 1}) \cdot \capacity_j > \cost_{\ell' + 1} \cdot (\capsum_{\ell'} - k), &j \geq \ell + 1
\end{cases}
\end{equation}}
\end{myproof}

\subsection{Proof of~\cref{prop:sybil-good}}
\aviv{TODO: go through}
\begin{myproof}
    Suppose $\cost_{j, 1} \le \cost_{j, 2} \le\cdots\le \cost_{j, r}$.
    After the Sybil attack, suppose the allocated number of provers is $\ell'$ and the allocated number of tasks is $\capsum'_{\ell'}$.
    Also suppose the new prover bids are $\cost'_1\le\cost'_2\le\cdots$ after the Sybil attack.

    If $j \leq \ell$, such Sybil attacks are never profitable.
    \begin{itemize}[leftmargin=*]
        \item If $\cost_{j, r} \le \cost_{\ell + 1}$, then $\ell' 
 = \ell +r -1$, thus prover $j$ still gets the same utility.
        \item If there exists $q\leq r - 1$ such that $\cost_{j, q} \leq \cost_{\ell + 1} < \cost_{j, q + 1}$, we show that all the splits with index greater than $q$ will not be allocated by contradiction.
        Assume by contradiction that the split with index $q + 1$ will be allocated.
        It follows that $\cost'_{\ell' + 1}\ge \cost_{\ell + 2}$.
        Recall that $\ell = \underset{k}{\arg\max} \{ \cost_{k + 1} \leq \txfee_{\capsum_k + 1} \}$, and $\ell' = \underset{k}{\arg\max} \{ \cost'_{k + 1} \leq \txfee_{\capsum'_k + 1} \}$.
        Therefore,
        $
        \cost_{\ell + 2} \le \cost'_{\ell' + 1} \le \txfee_{\capsum'_{\ell'} + 1} \le \txfee_{\capsum_{\ell} + 1}
        $.
        However, by the definition of $\ell$ we have $\cost_{\ell + 2} > \txfee_{\capsum_{\ell + 1} + 1}\ge \txfee_{\capsum_{\ell} + 1}$, which leads to a contradiction.
        Therefore, prover $j$ will at most gain utility $\sum_{k = 1}^q \capacity_{j, k}\cdot (\cost_{\ell + 1} - \cost_j)$, and will lose utility.
    \end{itemize}

    If $j \ge \ell + 1$, we show that if with the Sybil attack prover $j$ can gain profit, then $\capsum'_{\ell'} > \capsum_\ell$ and $\cost_{\ell + 1} < \cost'_{\ell' + 1}$.
    To profit, we must have the first split is allocated, since prover $j$ is originally not allocated, and that $\cost_j <\cost'_{\ell' + 1}$.
    Thus, provers $1$ to $\ell$ are still allocated, so $\capsum'_{\ell'} \ge \capsum_\ell + \capacity_{j, 1} >\capsum_\ell$.
    Note that $\cost_j \ge \cost_{\ell + 1}$, we have $\cost_{\ell + 1}\le \cost_j < \cost'_{\ell' + 1}$.
    Therefore, welfare will increase.
\end{myproof}

\subsection{Proof of \cref{prop:collusion-prover-user}}
\begin{myproof}
Suppose the coalition comprises a user with valuation $\txfee$, and a prover with valuation $\cost$ and capacity $\capacity$, and that as part of their collusion, the user bids $\txfee'$, and the prover bids a cost $\cost'$ and capacity $\capacity'$, where $(\txfee', \cost', \capacity') \neq (\txfee, \cost,\capacity)$.
If $\capacity' > \capacity$ then the prover will get slashed when it is allocated, so let $\capacity' \le \capacity$.
We now split our analysis into different cases.
\begin{itemize}[leftmargin=*]
    \item If $(\txfee', \cost') = (\txfee, \cost)$, then $\capacity' < \capacity$.
    Given $2\capacity$ tasks with fee $2\cost$ and another prover with cost $2\cost$ and capacity $\capacity$, then the colluders' joint utility is $\capacity \cdot \cost$, which is lower than they could have obtained honestly.
    \item If $\txfee' > \txfee$, we give an counterexample as follows.
    There are another two provers, both with $\capacity$ capacity, and with costs $0$ and $\txfee'$ respectively, and there are $3\capacity$ other tasks with fee $\txfee'$.
    If $\cost' = \cost$, then after the strategic bidding the joint utility will reduce at least $\txfee' - \txfee$.
    If $\cost' > \cost > \txfee'$ or $\txfee' > \cost' > \cost$, then after the strategic bidding the joint utility will reduce at least $\txfee' - \txfee$.
    If $\cost' > \txfee' > \cost$, then after the strategic bidding the joint utility will reduce at least $\txfee' - \txfee + \capacity' \cdot (\txfee' - \cost)$.
    If $\cost' < \cost \le \txfee'$ or $\txfee' \le \cost' < \cost$, then after the strategic bidding the joint utility will reduce at least $\txfee' - \txfee$.
    If $\cost' <\txfee' < \cost$, then after the strategic bidding the joint utility will reduce at least $\txfee' - \txfee + \capacity' \cdot (\cost - \txfee')$.
    \item If $\txfee' < \txfee$, we give an counterexample as follows.
    There are two other provers, both with $\capacity$ capacity, and with costs $0$ and $(\txfee + \txfee') / 2$ respectively.
    There are $3 \capacity$ other tasks with fee $(\txfee + \txfee') / 2$.
    If $\cost' = \cost$, then after the strategic bidding the joint utility will reduce at least $(\txfee- \txfee') / 2$.
    If $\cost' > \cost \ge (\txfee + \txfee') / 2$ or $(\txfee + \txfee') / 2 \ge \cost' > \cost$, then after the strategic bidding the joint utility will reduce at least $(\txfee- \txfee') / 2$.
    If $\cost' > (\txfee + \txfee') / 2 > \cost$, then after the strategic bidding the joint utility will reduce at least $\txfee' - \txfee + \capacity' \cdot ((\txfee + \txfee') / 2 - \cost)$.
    If $\cost' < \cost \le (\txfee + \txfee') / 2$ or $(\txfee + \txfee') / 2 \le \cost' < \cost$, then after the strategic bidding the joint utility will reduce at least $(\txfee' - \txfee) / 2$.
    If $\cost' <(\txfee + \txfee') / 2 < \cost$, then after the strategic bidding the joint utility will reduce at least $ \txfee' - \txfee + \capacity' \cdot (\cost - (\txfee + \txfee') / 2)$.
    \item If $\txfee' = \txfee$, then $\cost' \ne \cost$, and we give a counterexample as follows.
    There is another prover with $\capacity$ capacity and bid $(\cost + \cost') / 2$, and there are $2 \capacity$ tasks with fee $(\cost + \cost') / 2$.
    If $\cost' > \cost$, after the strategic bidding the joint utility will reduce at least $\capacity' \cdot (\cost' - \cost) / 2$.
    If $\cost' < \cost$, after the strategic bidding the joint utility will reduce at least $\capacity' \cdot (\cost - \cost') / 2$.
\end{itemize}
\end{myproof}

\subsection{Proof of \cref{prop:collusion-two-prover}}
\begin{myproof}
    Assume that the provers have capacities and costs $(\capacity_1, \cost_1)$ and $(\capacity_2, \cost_2)$ respectively with $\cost_1 < \cost_2$, and strategically bid $(\capacity'_1, \cost'_1)$ and $(\capacity'_2, \cost'_2)$, with $\capacity'_1 \le \capacity_1$ and $\capacity'_2 \le \capacity_2$.
    (Otherwise, if they bid high capacity they will be slashed when allocated.)
    We discuss the cases as follows:
    \begin{itemize}[leftmargin=*]
        \item If $\cost'_1 = \cost_1$ and $\cost'_2 = \cost'_2$, we give a counterexample as follows.
        There is another prover with cost $2\cost_2$ and capacity $1$, and there are $(\capacity_1 + \capacity_2 + 1)$ tasks with fee $2\cost_2$.
        Then their joint utility will decrease $\cost_2 \cdot (\capacity_2 - \capacity'_2) + (2\cost_2 - \cost_1) \cdot (\capacity_1 - \capacity'_1)$.
        \item If $\cost'_2 > \cost_2$, we give a counterexample as follows.
        There is another prover with cost $(\cost_2 + \cost'_2) / 2$ and capacity $\capacity_2$, and there are $\capacity_1 + 2\capacity_2$ tasks with fee $(\cost_2 + \cost'_2) / 2$.
        In this case, their joint utility will decrease at least $\capacity_2 \cdot (\cost'_2 - \cost_2) / 2$.
        \item If $\cost'_2 < \cost_2$, we give a counterexample as follows.
        There is another prover with cost $(\cost_2 + \cost'_2) / 2$ and capacity $\capacity_2$, and there are $\capacity_1 + 2\capacity_2$ tasks with fee $(\cost_2 + \cost'_2) / 2$.
        In this case, their joint utility will decrease at least $\capacity'_2 \cdot (\cost_2 - \cost'_2) / 2$.
        \item If $\cost'_1 < \cost_1$, we give a counterexample as follows.
        There is another prover with cost $(\cost'_1 + \cost_1) / 2$ and capacity $\capacity_1$, and there are $2\capacity_1 + \capacity_2$ tasks with fee $(\cost'_1 + \cost_1) / 2$.
        In this case, their joint utility will decrease at least $\capacity'_1 \cdot (\cost_1 - \cost'_1) / 2$.
        \item If $\cost'_1 > \cost_1$, we give a counterexample as follows.
        There is another prover with cost $(\cost'_1 + \cost_1) / 2$ and capacity $\capacity_1$, and there are $2\capacity_1 + \capacity_2$ tasks with fee $(\cost'_1 + \cost_1) / 2$.
        In this case, their joint utility will decrease at least $\capacity_1 \cdot (\cost'_1 - \cost_1) / 2$.
    \end{itemize}
\end{myproof}

\section{Efficiency of the Core Mechanism}
We define efficiency loss as the difference between the social welfare of the outcome of the mechanism and the optimal social welfare.
Following the results of~\cite{mcafee1992dominant}, if all provers have unit capacity, the core mechanism of \tool achieves welfare equal to $1/\min (\text{\# users}, \text{\# provers})$, which is asymptotically optimal.
It is shown in~\cite{segalhalevi2018muda} that for any strongly budget-balanced and incentive-compatible mechanism where a seller has $M$ items and a buyer wishes to purchase at most $M$ units, the efficiency asymptotically converges to $0$ in $M$ when the per-item trade price is determined exogenously.
We suspect that a similar result may apply to the core mechanism of \tool in a myopic setting, due to pathological cases in which the capacity of the first unallocated prover is large.
However, we hypothesize that a non-myopic analysis in which pent-up user demand ``spills'' over to the next rounds may show that these pathologies have limited impact in the long run, making such an analysis a promising direction for future work.
For example, the related work of \citeauthor{nisan2023serial}~\cite{nisan2023serial} shows the monopolistic TFM results in unbounded welfare losses in the myopic case, while in the non-myopic setting, the obtained welfare is at least $\frac{1}{2}$ of the optimal one.
Several related works analyze far-sighted actor incentives, such as \citeauthor{gafni2024scheduling}~\cite{gafni2024scheduling} who consider competitive allocation strategies, and \citeauthor{penna2024serial}~\cite{penna2024serial} who derive minimal ``admission'' prices that guarantee eventual transaction allocation.

\end{document}